\documentclass[journal, onecolumn]{IEEEtran}
\usepackage[utf8]{inputenc}
\usepackage[english]{babel}
\usepackage{cite}
\usepackage{graphicx}
 \graphicspath{{./figures/}}
\usepackage[justification=centering]{caption}
\usepackage{graphicx}
\usepackage{subcaption}
\usepackage{amsmath}
\usepackage{amsthm,amsfonts,amstext,amssymb}
\usepackage{enumitem}   
\usepackage{indentfirst}
\usepackage[font=small,skip=0pt]{caption}
\usepackage{float}
\usepackage{mathtools}
\usepackage{xcolor}
\usepackage{xr}
\usepackage{tikz}
\usetikzlibrary{shapes,arrows,positioning,calc}
\usepackage{standalone}
\usepackage[pdftex]{hyperref}
\usepackage[normalem]{ulem}

\newtheorem{theorem}{Theorem}

\theoremstyle{definition}

\theoremstyle{remark}
\newtheorem{remark}{Remark}

\begin{document}

\title{Vehicle Following On A Ring Road Under Safety Constraints: Role of Connectivity and Coordination}
\author{Milad~Pooladsanj,
        Ketan~Savla,
        and~Petros~A.~Ioannou
\thanks{M. Pooladsanj and P. A. Ioannou are with the Department
of Electrical Engineering, University of Southern California, Los Angeles, CA 90007 USA e-mail: (pooladsa@usc.edu; ioannou@usc.edu).}
\thanks{K. Savla is with the Sonny Astani Department of Civil and Environmental Engineering, University of Southern California, Los Angeles CA 90089 USA (e-mail: ksavla@usc.edu). K. Savla has financial interest in Xtelligent, Inc.} 
\thanks{
This work was supported in part by NSF CMMI 1636377 and METRANS 19-17.}
}\maketitle
\begin{abstract}
A fundamental problem in traffic networks is driving under safety and limited physical space constraints. In this paper, we design longitudinal vehicle controllers and study the dynamics of a system of homogeneous vehicles on a single-lane ring road in order to understand the interplay of limited space, speed, and safety. Each vehicle in the system either operates in the cruise control mode or follows a vehicle ahead by keeping a safe time headway. We show that if the number of vehicles is less than a certain critical threshold, vehicles can occupy the limited space in many different configurations, i.e., different platoons of different sizes, and they converge to a uniform maximum speed while attenuating errors in the relative spacing upstream a platoon. If the number of vehicles exceeds the threshold, vehicles converge to a unique symmetric configuration and the equilibrium speed decreases as the number of vehicles increases. Next, we consider vehicle-to-vehicle (V2V) communication and show that it increases the critical number of vehicles that can travel with the maximum speed. Finally, we consider central coordination and show that the proposed controllers can force vehicles to converge to a desired configuration specified by the coordinator while maintaining safety and comfort. We demonstrate the performance of the proposed controllers via simulation.
\end{abstract}
\section{Introduction}\label{Section: Introduction}
Traffic congestion has costed billions of dollars, hours, and gallons of fuel in the past years \cite{schrank2015}. While this congestion is a direct consequence of high travel demand competing to utilize the limited supply of road networks in a safe manner \cite{ALASIRI2020, varaiya1993smart}, it is magnified by poor human drivers' response to various disturbance \cite{Malikopoulos.2017}. It has been reported that Connected and Autonomous Vehicles (CAVs) have the potential to compensate for human errors and corresponding delays to effectively improve the throughput and capacity of highways \cite{Malikopoulos.2017}. However, the analysis of the limited capacity of road networks, which we shall refer to as \emph{bounded space}, in conjunction with safety constraints has received little attention in microscopic traffic studies of CAVs. 
\par
The impact of CAVs on traffic flow in the longitudinal direction is often evaluated by considering autonomous vehicles on an unbounded single lane road with no passing. The majority of research in this direction consider a platoon of vehicles and assume that the leader of the platoon follows a desired speed trajectory. The objective is then to design state-feedback throttle/brake controllers for the following vehicles such that they can adjust their speed to the speed of the leader while keeping a safe distance from the next vehicle \cite{Ioannou.Xu.1994, Ioannou.Chien.1993}. The dynamical analysis of the following vehicles provide results on collision avoidance, attenuation of errors upstream the platoon, the effect of delay in the system performance, ride comfort, and the impacts of integration of communication channels \cite{Ioannou.Xu.1994,Ioannou.Chien.1993,Swaroop.Hedrick.1999,Shaw.Hedrick.2007,Tan.Rajamani.Zhang.1998,Seiler.Pant.Hedrick.2004,
Barooah.Mehta.Hespanha.2009,Zhang.Kosmatopoulos.Ioannou.Chien.1999,Lin.Fardad.Jovanovic.2012,Chu.1974,Seiler.2001}. The aforementioned research efforts aim their attention at evaluating the performance of CAVs when they are already in the vehicle following mode. In practice, the performance of CAVs in handling different situations such as switching between different modes of operation must be also taken into account. In \cite{raza1996vehicle}, a supervisory controller was designed and analyzed which was responsible for interacting with the throttle/brake controller, choosing the proper mode of operation, e.g., cruise control or vehicle following, and the transition between these modes, and detecting any irregular behavior such as an emergency stopping situation. In all of these studies, however, increasing the number of vehicles does not affect the speed or density since an infinite space is assumed. In other words, the aforementioned analyses describe, at best, the behavior of CAVs under the safety constraint but not the space constraint.
\par
The analytical understanding of the interplay between the safety constraint, speed limit, and space limitation by using a simple road geometry will help analyze these effects for more complicated road geometries and networks where space is limited. A simple, but practical, such setup is a ring road. In the experiment described in \cite{Sugiyama2008TrafficJW}, a single lane ring road was used to show the formation of stop-and-go waves when all vehicles are human-driven and there is no bottleneck. Inspired in part by \cite{Sugiyama2008TrafficJW}, there has recently been dynamical analysis on this setup for mixed-autonomy settings \cite{Cui2017, zheng2018smoothing, stern2018dissipation}. The foci of the analytical aspects of these works, however, is on the formation and dissipation of traffic jams using autonomous vehicles for a high density scenario, without explicit consideration of safety. Earlier work from the authors \cite{pooladsanj2020vehicle} explicitly included the safety constraint for a simple vehicle model. However, the impact of V2V communication and central coordination of vehicles on the bounded space was not addressed. Practical bounds on the acceleration of vehicles was not considered either.
\par
In this paper, we consider homogeneous automated vehicles on a closed single lane ring road. We adopt a nonlinear vehicle model with first-order engine dynamics from \cite{CAUDILL1976195} derived from Newton's second law of motion. We design state-feedback control laws based on feedback linearization to control the throttle and brake commands. Three different scenarios are considered. In the first scenario, we assume that vehicles do not communicate and there is no coordination. Each vehicle either follows a constant speed trajectory, i.e., is in the cruise control mode, or safely follows the vehicle ahead, i.e., the vehicle following mode, by keeping a safe headway. Transitioning between modes of operation is determined by a combination of relative spacing and speed signals and is handled by the vehicle's supervisory controller as discussed in \cite{raza1996vehicle}. It is analytically shown that the equilibrium of this dynamical system leads to the well-known triangular fundamental diagram. In other words, the interplay of bounded space, speed limit, and safety can be quantified in a straightforward manner with this problem formulation. We explicitly characterize the critical density $\rho_c$ of the fundamental diagram, at which the flow is maximized, in terms of system parameters (time headway constant, free flow speed, minimum standstill safety distance, and length of each vehicle). 
In the second scenario, we assume that vehicles are also able to communicate their braking capabilities and their instantaneous acceleration with their immediate predecessor. We show that V2V communication increases $\rho_c$ and the capacity of the road compared to the first scenario, thus it enhances the free flow region of the fundamental diagram. In the final scenario, we assume that a central coordinator communicates the desired platoon formations and inter-platoon spacings to certain vehicles. 
We prove that the designed controllers guarantee robust speed tracking and/or vehicle following, and attenuation of errors in the desired relative spacing, speed, and acceleration upstream a platoon while satisfying desired acceleration bounds for all three scenarios. We demonstrate the performance of the controllers by simulating different scenarios.
\par
\par
The contributions of this paper can be summarized as follows:\par
\begin{enumerate}
    \item We illustrate the utility of a ring road setup to analyze the impact of space limitation, speed limit, safety constraint, and V2V communication on flow and density in a straightforward manner;
    \item We show that when every vehicle seeks to attain the maximum possible speed while respecting speed limit and safe distance to the vehicle in front, then the emergent configuration of inter-vehicle spacing is unique at high density but not at low density; 
    \item We describe a protocol by which a central coordinator can safely achieve a desired configuration at low density.
\end{enumerate}
\par
The rest of the paper is outlined as follows. In section \ref{Section: Problem Formulation}, we state the problem formulation and control objectives. In Section \ref{Section: No Coordination}, we first consider the case where vehicles do not communicate and there is no coordination on the ring road. We next extend the analysis for the case where V2V communication is possible. In Section \ref{Section: Coordination}, the role of central coordination on the ring road is evaluated and suitable transition logic are proposed. Section \ref{Section: Simulation} provides simulation results for these scenarios. We conclude the paper and discuss future directions in section \ref{Section: Future Work}.

\section{Problem Formulation}\label{Section: Problem Formulation}
\subsection{Basic Notations}\label{Subsection: Basic Notations}
Consider $n$ homogeneous vehicles of length $L$, on a closed ring road. Without loss of generality, assume that the perimeter of the ring road is $P$ for some $P > nL$. We assign coordinates over the distance interval $[0,P]$ to the ring road in the clock-wise direction. Let $\mathcal{N} = \{1, 2, \cdots, n\}$ be the set of vehicles' indices, where vehicle $i$ is the $i^{th}$-closest vehicle to point $0$ at time $t = 0$. Let $0 \leq x_i(t) < \infty$ denote the distance traveled by the $i^{th}$ vehicle with respect to a fixed reference point on the roadside (without loss of generality we assume that this reference point is the point $0$), and $v_i(t)$, $a_i(t)$ denote the speed and acceleration at time $t \geq 0$, respectively. Moreover, let $y_i(t) = x_{i+1}(t) - x_i(t) - L$ be the relative spacing of the $i^{th}$ vehicle with respect to the vehicle $i+1$ ahead at time $t \geq 0$, where $x_{n+1} \equiv P + x_1$ due to the periodicity of the ring road. Throughout the paper, except when needed, we use $x_i$, $v_i$, $a_i$, and $y_i$ without explicitly mentioning their dependence on time. For simplicity of notations, we formulate the vehicle model and controller design for an ego vehicle with subscript $e$ and use the subscript $l$ in order to differentiate between the ego vehicle and its vehicle ahead, i.e., the lead vehicle.
%
Note that by definition, $\sum_{i = 1}^{n} y_i = P - nL$. This constraint is the main contrast to a straight line with no space limitation. An illustration of this setup for three vehicles is depicted in Figure \ref{Fig: config}.  
\par
 \begin{figure}[H]
\centering
\includegraphics[width=0.4\textwidth]{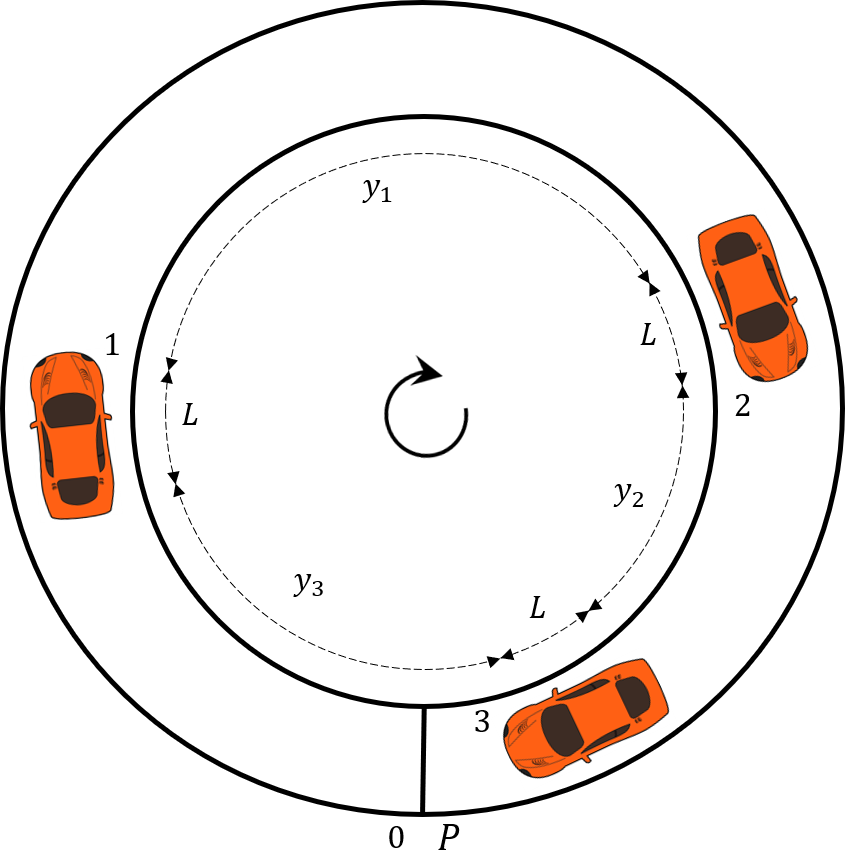}
\caption{\sf Example of three vehicles on a closed ring road setup}\label{Fig: config}
\end{figure}
 \begin{figure}[th]
\centering
\includegraphics[width=0.6\textwidth]{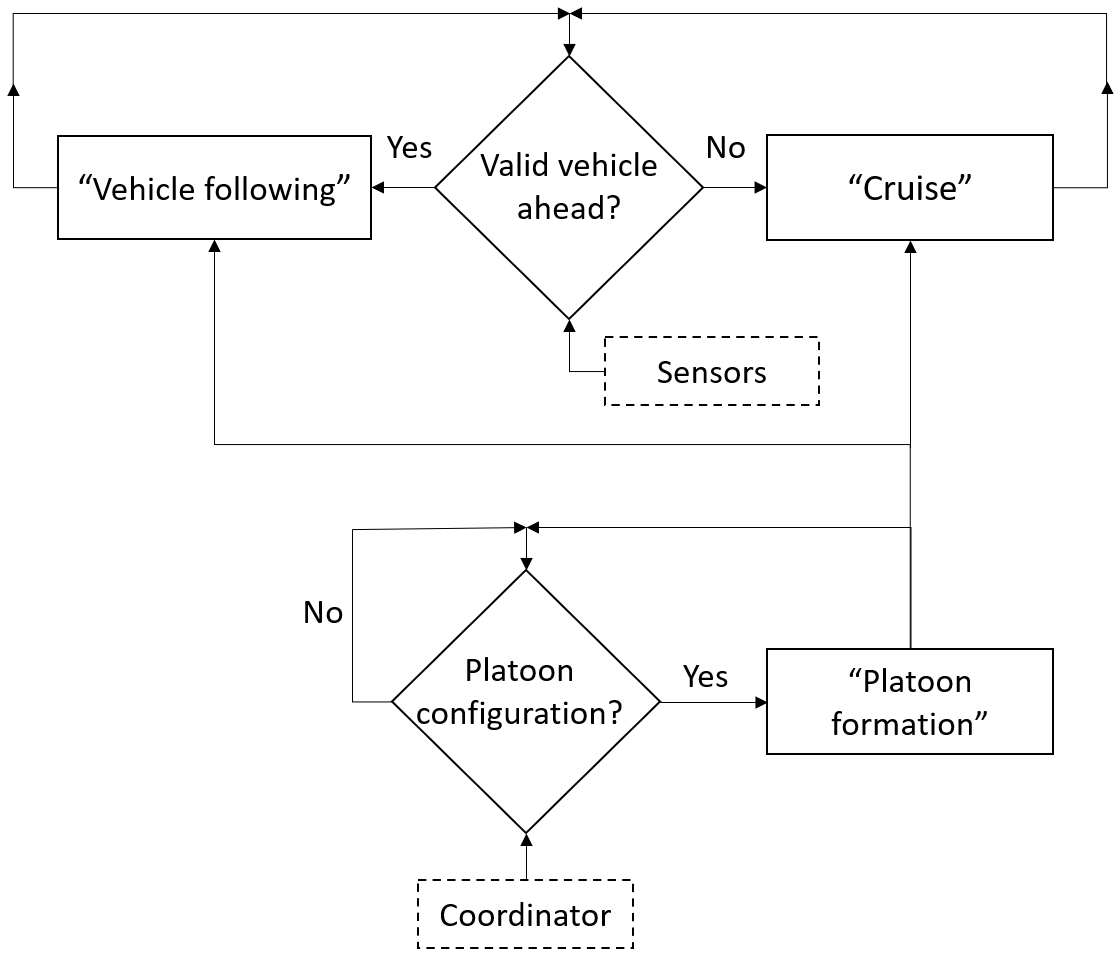}
\vspace{0.2 cm}
\caption{\sf Logic diagram for determining the mode of operation}\label{Fig: Logic for Mod of Operation}
\end{figure}
\subsection{Modes of Operation}\label{Subsection: Mode of Operation Selection}
Each vehicle operates in one of the following two modes of operation: cruise control or vehicle following, see Figure \ref{Fig: Logic for Mod of Operation}. If there is no central coordination, an ego vehicle operates in the cruise mode if no valid vehicle is ahead that is within its sensing range. The validity of the lead vehicle is determined by comparing the relative spacing to a design threshold value. The ego vehicle is in the vehicle following mode, i.e., it follows the lead vehicle by keeping a safety distance, as long as the lead vehicle's speed is within the allowable speed limit. The speed limit is taken to be equal to the free flow speed $V_f$ when there is no coordination. The platoon formation state in Figure \ref{Fig: Logic for Mod of Operation} is activated in order to achieve a desired platoon formation when a central coordinator is present. This state will be discussed in detail in Section \ref{Section: Coordination}. \par
\subsection{Vehicle Model}\label{Subsection: Vehicle Model}
We assume that the road surface is horizontal and there is no wind gust. We use Newton's second law of motion for the ego vehicle to write,
\begin{equation*}\label{Eq: (Acceleration Control)Newton's second law of motion}
  m_ea_e = F_e - k_{d}v^2_e - d_{m}(v_e)
\end{equation*}
where $F_e$ is the engine force, $m_e$ is the mass, $k_{d}$ is the aerodynamic drag coefficient, and $d_{m}(v_e)$ is the mechanical friction of the ego vehicle travelling with the speed $v_e$ \cite{CAUDILL1976195}. Assuming a first-order engine dynamics we have,
\begin{equation*}\label{Eq: (Acceleration Control)Engine dynamics}
  \dot{F}_e = \frac{1}{\tau(v_e)}(\theta_e - F_e)
\end{equation*}
where $\theta_e$ is the throttle angle's force to the engine, and $\tau(v_e)$ is the engine's time constant at the speed $v_i$ \cite{Sheikholeslam:M91/115}.

By combining the last two equations we derive,
\begin{equation*} \label{Eq: Vehicle model} 
  \dot{a}_e = \beta(v_e,a_e) + \alpha(v_e)\theta_e 
\end{equation*}
where,
\begin{equation*}\label{Eq: (Acceleration Control)alpha(v_i)}
\begin{split}
   \alpha(v_e)  & = \frac{1}{m_e \tau(v_e)} \\
    \beta(v_e,a_e)  & = -2\frac{k_{d}}{m_e}v_e a_e - \frac{1}{m_e}\dot{d}_{m}(v_e) - \frac{1}{\tau(v_e)}[a_e + \frac{k_{d}}{m_e} v^2_e + \frac{d_{m}(v_e)}{m_e}]
\end{split}
\end{equation*}
\par
At each speed $v_e$, throttle angle is chosen such that,
\begin{equation*}\label{Eq: (Acceleration Control)(Time) nonlinear feedback of u_i}
    \theta_e = \frac{1}{\alpha(v_e)}[u_e - \beta(v_e,a_e)]
\end{equation*}
which leads to the equation,
\begin{equation} \label{Eq: Simplified Vehicle dynamics} 
\dot{a}_e = u_e 
\end{equation}
where $u_e$ is to be designed to meet the control objectives presented below: 
\begin{enumerate}
    \item  \textbf{Safety}: no rear-end collision under a worst-case stopping scenario as explained in  \cite{Ioannou.Chien.1993} 
    \item \textbf{Smooth longitudinal maneuver}: Smooth position and/or speed tracking in the two modes of operation as well as a smooth transition between these modes
    \item \textbf{String error attenuation}: attenuation of the amplitude of errors, e.g., in the position, upstream a platoon
    \item \textbf{Passenger comfort}: $a_{min} \leq a_e \leq a_{max}$, except in an emergency braking scenario, and small jerk $\dot{a}_e$ \cite{Ioannou.Xu.1994} 
\end{enumerate}
\par
In the following sections, we design and analyze control laws that can meet the objectives with and without V2V communication and in the presence of a central coordinator.
\section{Vehicles On a Ring Road Without Coordination}\label{Section: No Coordination}
\subsection{No V2V Communication}\label{SubSection: No Coordination, controller design}
In this section, we assume that vehicles do not communicate with each other and obtain the necessary data for cruising or vehicle following by using their own sensing capabilities. When the mode of operation is determined as explained in Section \ref{Subsection: Mode of Operation Selection}, the sensing data are passed through appropriate filters \cite{raza1996vehicle} in order to generate continuous-time signals passed to the longitudinal controller $u_e$ designed as follows:
\begin{enumerate}
    \item \textbf{Cruise}: \par
    \begin{align}
    u_e &= K_a a_e + C_{v}(v_{r} - v_e) 
    + \int_{0}^{t}[C_{s}(v_{r} - v_e)]d\tau \label{Eq: (Acceleration Control) Design of c_i}\\
    \dot{v}_r &= \text{sat}[p(V_s - v_{r})],~ v_r(0) = v_e(0) \label{Eq: (Acceleration Control) Acceleration limiter} \\
        \text{sat}[x] &= \begin{cases*}
       a_{max} & \mbox{if } $x \geq  a_{max}$ \\
       x & \mbox{if } $a_{min} < x < a_{max}$ \\
       a_{min} & \mbox{if } $x \leq a_{min}$
   \end{cases*} \label{Eq: (Acceleration Control) Saturation function definition}
\end{align}
    \item \textbf{Vehicle following}: \par
    \begin{align}
   u_e &= K_a a_e + C_{p}(t)\delta_e + C_{v}(v_{r} - v_e) 
    + \int_{0}^{t}[C_{q}(\tau)
    \delta_{e} + C_{s}(v_{r} - v_e)]d\tau \label{Eq: (Acceleration Control) Design of c_i when B_F = 1}\\
     v_r &= v_l + (v_r(0) - v_l)e^{-\lambda t} \label{Eq: (Acceleration Control) reference speed when B_F = 1} \\
        \delta_e &= y_e - (hv_e + S_0) \label{Eq: (Acceleration Control) Reference spacing when B_F = 1}
\end{align}
\end{enumerate}
where $K_a < 0$, $C_v, C_s, p, \lambda > 0$ are design constants, $V_s$ is the speed limit and $V_s = V_f$.
Moreover, the threshold distance $\Delta_d$ for switching from the cruise control mode to the vehicle following mode is chosen as follows,
\begin{equation}\label{Eq: Switching distance Delta_d}
    \Delta_d = \begin{cases*}
            hv_e + S_0 + r(v_{e} - v_l) & \mbox{if } $v_e \geq v_l$ \\
            hv_e + S_0 & \mbox{otherwise }
                \end{cases*}
\end{equation}
where $r > 0$ is a design constant. If the relative spacing of the ego vehicle with respect to the lead vehicle ahead is greater than $\Delta_d$ at $t = 0$, the ego vehicle starts operating in the cruise control mode and the speed tracking controller \eqref{Eq: (Acceleration Control) Design of c_i} is used. The reference speed $v_r$ in this case is generated by passing the desired speed limit $V_s$ through the nonlinear acceleration limiter filter \eqref{Eq: (Acceleration Control) Acceleration limiter} with the saturation function described in \eqref{Eq: (Acceleration Control) Saturation function definition}. The acceleration limiter prevents the acceleration outside the comfortable range when there is a large initial speed error $V_s - v_e(0)$ \cite{Ioannou.Xu.1994}. If the relative spacing becomes less than $\Delta_d$ at some time $t_0 \geq 0$, the ego vehicle switches to the vehicle following mode and the speed/position tracking controller \eqref{Eq: (Acceleration Control) Design of c_i when B_F = 1} is used. The design parameters $C_p(t), C_q(t)$ are smoothly increased from zero to some positive design constants $C_p, C_q > 0$, i.e.,
$C_p(t) = C_p(1 - e^{-\lambda(t - t_0)})$, $C_q(t) = C_q(1 - e^{-\lambda(t - t_0)})$, $t \geq t_0$. Moreover, the reference speed $v_r$ is smoothly changed from the initial value to the speed of the lead vehicle $v_l$ (see \eqref{Eq: (Acceleration Control) reference speed when B_F = 1}), and the reference relative spacing is set to $hv_e + S_0$ (see \eqref{Eq: (Acceleration Control) Reference spacing when B_F = 1}), where $S_0 > 0$ is a constant standstill separation distance and $h$ is a safe time headway constant. This is a well-known safe vehicle following strategy where the following vehicles try to keep a safe constant time headway from the vehicle ahead \cite{Ioannou.Chien.1993}. It was shown that the value of the time headway constant can be chosen such that two consecutive vehicles do not collide under a worst-case stopping scenario as explained in \cite{Ioannou.Chien.1993}. Accordingly, the switching distance $\Delta_d$ is chosen to be equal to the safety distance $hv_e + S_0$ plus an additional non-negative term $r(v_e - v_l)$ if the ego vehicle is travelling at least as fast as the lead vehicle (see \eqref{Eq: Switching distance Delta_d}). The ego vehicle keeps operating in the vehicle following mode as long as the lead vehicle's speed is within its allowable speed limit $V_s$. 
\par
\begin{remark}
The objective is to design the control parameters such that \eqref{Eq: (Acceleration Control) Design of c_i} - \eqref{Eq: Switching distance Delta_d} ensures stability, string error attenuation, and, lastly, satisfies comfort. We should emphasise that the designed longitudinal controller is only responsible for smoothly adjusting the spacing and/or speed. Other operations such as emergency braking are assessed by a higher-level supervisory controller and operated by different control laws which are not addressed in this paper. However, this problem is resolved in other papers, see for example \cite{Ioannou.Xu.1994},\cite{raza1996vehicle}. 
\end{remark}
\par
We define $n_c = \frac{P}{hV_f + S_0 + L}$ as the \emph{critical number of vehicles} on the ring road ($n_c$ can be non-integer). We also define \emph{configuration} as the vector of relative spacings on the ring road.
\begin{theorem}\label{Theorem: (Acceleration Control) Robust veh following and speed tracking} There exist design parameters such that the following hold,
\begin{enumerate}[label=(\roman*)]
    \item The controller \eqref{Eq: (Acceleration Control) Design of c_i} - \eqref{Eq: Switching distance Delta_d} guarantees smooth vehicle following and/or speed tracking in all modes of operation and attenuation of the amplitude of errors with respect to the desired relative spacing, speed, and acceleration upstream a platoon.
    \item If  $n < n_c$, there is an infinite number of vehicle configurations on the ring road; however the equilibrium speed is $V_f$ in each of these configurations.
    \item If $n \geq n_c$, there is a unique vehicle configuration where all vehicles are symmetrically distributed around the ring road and their speed converges to an equilibrium speed of $\frac{1}{h}(\frac{P}{n} - S_0 - L) \leq V_f$.
\end{enumerate}
\end{theorem}
\begin{proof}
Refer to Appendix \ref{Section: Appendix A}.
\end{proof}
\begin{remark}
Equations \eqref{Eq: (Appx A) Cruise relation for acceleration in time}, \eqref{Eq: (Appx A) VehFol relation for acceleration in time} in the proof of Theorem \ref{Theorem: (Acceleration Control) Robust veh following and speed tracking} suggest that when in the cruise control mode, a vehicle satisfies the comfortable acceleration limits and when in the vehicle following mode, it accelerates/decelerates at most as high as the vehicle ahead except, maybe, for an exponentially vanishing term. We confirm via simulations that this guarantees the specified comfort requirements except, maybe, for an exponentially vanishing time.
\end{remark}
\begin{remark}
According to Theorem \ref{Theorem: (Acceleration Control) Robust veh following and speed tracking}, for a given number of vehicles $n < n_c$, vehicles can form platoons of (possibly) different sizes with different inter-platoon spacing at steady state, which depends on the initial condition. We discuss in Section \ref{Section: Coordination} the role of central coordination in achieving a unique desired configuration in order to improve efficiency in utilizing the limited space.
\end{remark}
\begin{remark}\label{Remark: FD} Macroscopic traffic flow interpretation of Theorem \ref{Theorem: (Acceleration Control) Robust veh following and speed tracking}:
Let $v^{*}$ be the equilibrium speed of vehicles, $\rho = \frac{n}{P}$ be the space-mean \emph{density}, $\rho_c = \frac{n_c}{P}$ be the \emph{critical density}, and $q^{*} = \rho v^{*}$ be the equilibrium space-mean \emph{flow}. It follows from Theorem \ref{Theorem: (Acceleration Control) Robust veh following and speed tracking} that $v^{*} = \min \{ V_f, \frac{1}{h}(\frac{P}{n} - S_0 - L)\}$. Therefore,
\begin{equation*}
    q^{*} = \begin{cases*}
    V_f\rho & \mbox{if } $\rho < \rho_c$ \\
    \frac{1}{h}(1 - \rho(S_0 + L))  & \mbox{if } $\rho \geq \rho_c$
    \end{cases*}
\end{equation*}
\par 
In other words, when the density is less than the critical density, the flow increases linearly with increasing density. However, when the density exceeds the critical density, the flow decreases linearly with increasing density. This gives rise to the well-known triangular fundamental diagram (see Figure \ref{Fig: Fundamental Diagram}). The maximum value of $q^{*}$, i.e., the \emph{capacity} $C$ of the ring road, is then found to be $C = \frac{V_f}{hV_f + S_0 + L}$. 
\end{remark}
\begin{figure}[t]
    \centering
    \includegraphics[width=0.35\textwidth]{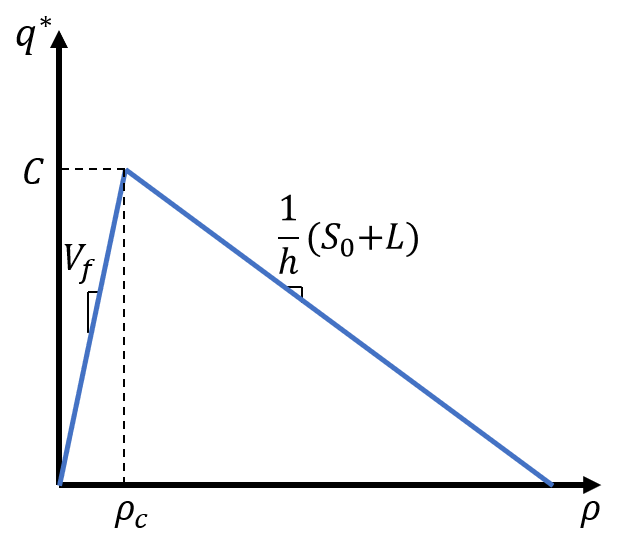}
    \caption{\sf Fundamental diagram without V2V communication}
    \label{Fig: Fundamental Diagram}
\end{figure}
\begin{remark}\label{Remark: (Acceleration Control) What design parameters should we choose}
According to the proof of Theorem \ref{Theorem: (Acceleration Control) Robust veh following and speed tracking}, for speed tracking in the cruise control mode the poles of $K(s)$ in \eqref{Eq: Tf K(s)} must lie in the open left half of the $s$-plane. This condition is satisfied if,
\begin{equation}\label{Eq: design param conditions to guarantee cruising}
        K_aC_v + C_s < 0
\end{equation}
\par
Moreover, for position/speed tracking and string error attenuation in the vehicle following mode, the design parameters must be chosen such that poles of $G(s)$ in \eqref{Eq: (Acceleration Control) G(s)} have negative real parts and $|G(j\omega)| \leq 1$, $\forall \omega \geq 0$. The former can be guaranteed by using pole placement. Additionally, $|G(j\omega)| \leq 1$, $\forall \omega \geq 0$ is satisfied if, 
\begin{equation}\label{Eq: (Acceleration Control) Conditions to gurantee stability on the ring road}
   \begin{aligned}
       & C_1 \geq 0 \\
       & C_2 - C^2_v \geq 0
   \end{aligned}
 \end{equation}
where,
\begin{equation*}
    \begin{aligned}
    C_1 &= K^2_a - 2(hC_{p} + C_{v}) \\
    C_2 &= (hC_p + C_v)^2 + 2C_q + 2K_a(C_p + hC_q + C_s) \\
    \end{aligned}
\end{equation*}
\par
We provide a set of parameters in Section \ref{Section: Simulation} that satisfies \eqref{Eq: design param conditions to guarantee cruising} and \eqref{Eq: (Acceleration Control) Conditions to gurantee stability on the ring road} (refer to \eqref{Eq: (Simulation)(Acceleration Control) Conditions to gurantee stability and string attenuation}) as well as the stability criterion for $G(s)$.
\end{remark}

\subsection{V2V Communication}\label{SubSection: (COMM) No Coordination, controller design}
We now assume that vehicles are able to communicate their braking capabilities as well as their instantaneous acceleration and deceleration to their immediate predecessor. This feature allows for accurate reference tracking when vehicles are outside the sensing range and also smaller safe time headway constant between vehicles \cite{milanes2013cooperative}. With V2V communication, the longitudinal control law in the vehicle following mode \eqref{Eq: (Acceleration Control) Design of c_i when B_F = 1} is modified as follows,
\begin{equation}\label{Eq: (COMM)(Acceleration Control) Design of c_i}
\begin{aligned}
    u_e &= K_a a_e + C_{p}(t)\delta_{e} + C_{v}(v_{r} - v_e) + C_a(t)(a_{l} - a_e) \\
    & + \int_{0}^{t}[C_{q}(\tau)\delta_{e} + C_{s}(v_{r} - v_e) + C_b(\tau)(a_{l} - a_e)]d\tau \\
\end{aligned}
\end{equation}
where $C_a(t), C_b(t) \geq 0$ are additional control parameters which behave similar to $C_p(t), C_q(t)$.  
Note that the only difference between \eqref{Eq: (Acceleration Control) Design of c_i when B_F = 1} and \eqref{Eq: (COMM)(Acceleration Control) Design of c_i} is the additional acceleration terms $C_a(t)(a_{l} - a_e)$ and $C_b(t)(a_{l} - a_e)$ in \eqref{Eq: (COMM)(Acceleration Control) Design of c_i}. Since by choosing $C_a(t) = C_b(t) = 0$, $\forall t \geq 0$, \eqref{Eq: (COMM)(Acceleration Control) Design of c_i} becomes identical to the control law in \eqref{Eq: (Acceleration Control) Design of c_i when B_F = 1} all of the results for stability, string error attenuation, and comfort holds when V2V communication is possible. In fact, V2V communication adds additional degrees of freedom in choosing the design constants in order to guarantee good tracking performance. \par
As mentioned earlier, V2V communication reduces the minimum safe time headway constant $h$. Thus, the critical number of vehicle $n_c$ for which vehicles can operate at the free flow speed increases. As a result, the critical density $\rho_{c} = \frac{1}{hV_f + S_0 + L}$ and the capacity $C = \frac{V_f}{hV_f + S_o + L}$ in Remark \ref{Remark: FD} are increased. Therefore, V2V communication expands the free-flow region of the Fundamental diagram, see Figure \ref{Fig: FD V2V vs. no V2V}. \par
Furthermore, using V2V communication, vehicles can be organized in platoons and decide among themselves certain configurations. Moreover, it allows for accurate tracking of the position, speed, and acceleration of vehicles ahead even when they are outside the sensing range. This feature expands the number of possible configurations that can be achieved in the presence of a coordinator. We discuss this in the next section.
\begin{figure}[t]
    \centering
    \includegraphics[width = 0.35\textwidth]{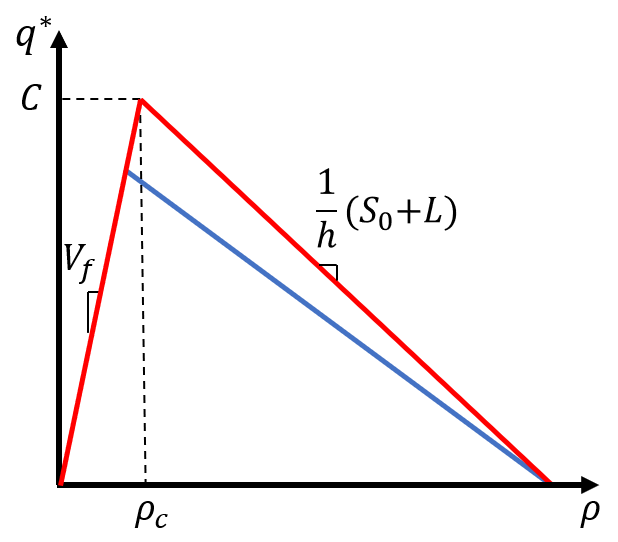}
    \caption{\sf Fundamental diagrams with V2V communication (red) and without V2V communication (blue)}
    \label{Fig: FD V2V vs. no V2V}
\end{figure}
\section{Coordination of Vehicles On A Ring Road}\label{Section: Coordination}
In the previous section, we assumed that vehicles travel without coordination, i.e., their action to achieve the speed limit or follow a vehicle in front was determined by their own sensors and/or V2V communication. According to Theorem \ref{Theorem: (Acceleration Control) Robust veh following and speed tracking}, if $n < n_c$, there is an infinite number of configurations in which the system of vehicles can occupy the limited space but, in all of them, they travel with the free flow speed $V_f$.
It may so happen that some configurations on the road are more desirable than others from the point of view of a central coordinator. For example, vehicles may be organized in closed-space platoons in order to decrease air drag and thus fuel consumption, or use the bandwidth of the coordinator-to-vehicle communication system more effectively by only communicating to the leaders of platoons \cite{raza1996vehicle}.  \par
In this section, we assume that vehicles use their own sensors and/or V2V communication in combination with commands from a central coordinator, see Figure \ref{Fig: Perf req of the controller}. The coordinator chooses a configuration from the following set of desired configurations and communicates it along with the perimeter of the road $P$, and the number of vehicles $n$ to certain vehicles:
\begin{enumerate}
    \item \textbf{$1$-platoon asymmetrical}: a single platoon of $n$ vehicles
    \item \textbf{Symmetrical}: all vehicles sharing the limited space equally
    \item \textbf{$m$-platoon symmetrical}:  $m$ platoons of vehicles, $1 < m \leq \frac{n}{2}$, sharing the limited space equally
\end{enumerate}
\begin{figure}[t]
    \centering
    \includegraphics[width = 0.55\textwidth]{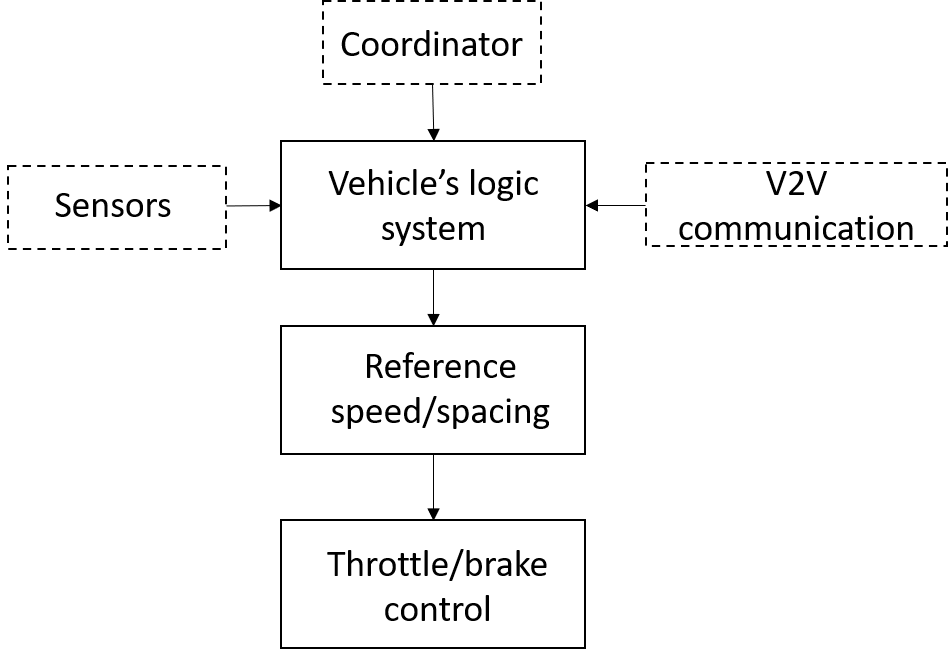}
    \vspace{0.2 cm}
    \caption{\sf Control structure in the presence of a coordinator and/or V2V communication}
    \label{Fig: Perf req of the controller}
\end{figure}
\par
These configurations are chosen in order to illustrate the idea. The following discussion can be easily applied to other desired configurations as well. Upon receiving commands from the coordinator, the ego vehicle calculates the appropriate reference speed and spacing according to the desired configuration and passes it to the longitudinal controllers $u_e$ in \eqref{Eq: (Acceleration Control) Design of c_i}, \eqref{Eq: (Acceleration Control) Design of c_i when B_F = 1} with the reference speed and spacing modified as follows,
\begin{align}
    V_s &= \begin{cases*}
    V_f & \mbox{if } desired platoon formed behind \\
    \alpha V_f & \mbox{otherwise}
    \end{cases*} \label{Eq: Effect of logic B_V on ref speed}\\
    \delta_e &= y_e - y_d \nonumber \\
    y_d &= \begin{cases*}
    (h_d + (h - h_d)e^{-\lambda t})v_e + S_0 & \mbox{if } $\begin{aligned}
        & \text{desired configuration requires} \\
        & \text{spacing adjustment}
    \end{aligned}$  \\
    hv_e + S_0 &  \mbox{otherwise }
    \end{cases*} \label{Eq: Effect of logic B_D on ref spacing}
\end{align}
\par
Upon receiving initiation commands from the coordinator at $t = 0$, the ego vehicle uses the platoon formation flow chart in Figure \ref{Fig: Speed selection logic} in order to calculate the reference speed and spacing. According to Figure \ref{Fig: Speed selection logic}, when the platoon for which the ego vehicle is its desired leader has not yet formed, the ego vehicle changes its speed limit $V_s$ from $V_f$ to $\alpha V_f$ for some $\alpha \in (0,1)$ (see \eqref{Eq: Effect of logic B_V on ref speed}). According to the switching logic explained in Section \ref{Subsection: Mode of Operation Selection}, it starts to decelerate, when it is safe, so that the vehicles behind catch up.
When the desired platoon has formed, the ego vehicle is notified by the coordinator and/or V2V communication and the speed limit is reset to $V_f$ (see \eqref{Eq: Effect of logic B_V on ref speed}). Moreover, if the ego vehicle needs to adjust its relative spacing depending on the desired configuration, the controller smoothly tracks $v_l$ and changes the reference relative spacing from the initial value to $h_dv_e + S_0$ (see \eqref{Eq: Effect of logic B_D on ref spacing}). The reference time headway constant $h_d$ is calculated by the ego vehicle such that $h_dV_f + S_0$ is equal to the reference relative spacing, e.g., when the desired configuration is symmetrical, the reference relative spacing is $\frac{P}{n} - L$, and $h_d = \frac{1}{V_f}(\frac{P}{n} - L - S_0)$. Since the reference time headway constant $h_d$ is different than $h$, the design parameters $C_p(t), C_q(t)$ in \eqref{Eq: (Acceleration Control) Design of c_i} are also smoothly changed, if necessary, such that the controller maintains good tracking performance. \par 
As an example, consider the $1$-platoon asymmetrical desired configuration and let the ego vehicle be its desired leader. Then, if the desired platoon has not yet formed, the ego vehicle switches to the cruise control mode, when it is safe, in order to track the reference speed of $\alpha V_f$, $\alpha \in (0,1)$, until all other vehicles catch up and switch to the vehicle following mode. At this point, it starts tracking the reference free flow speed $V_f$ and the transition is completed.
\begin{figure}[t]
    \centering
    \includegraphics[width = 0.65\textwidth]{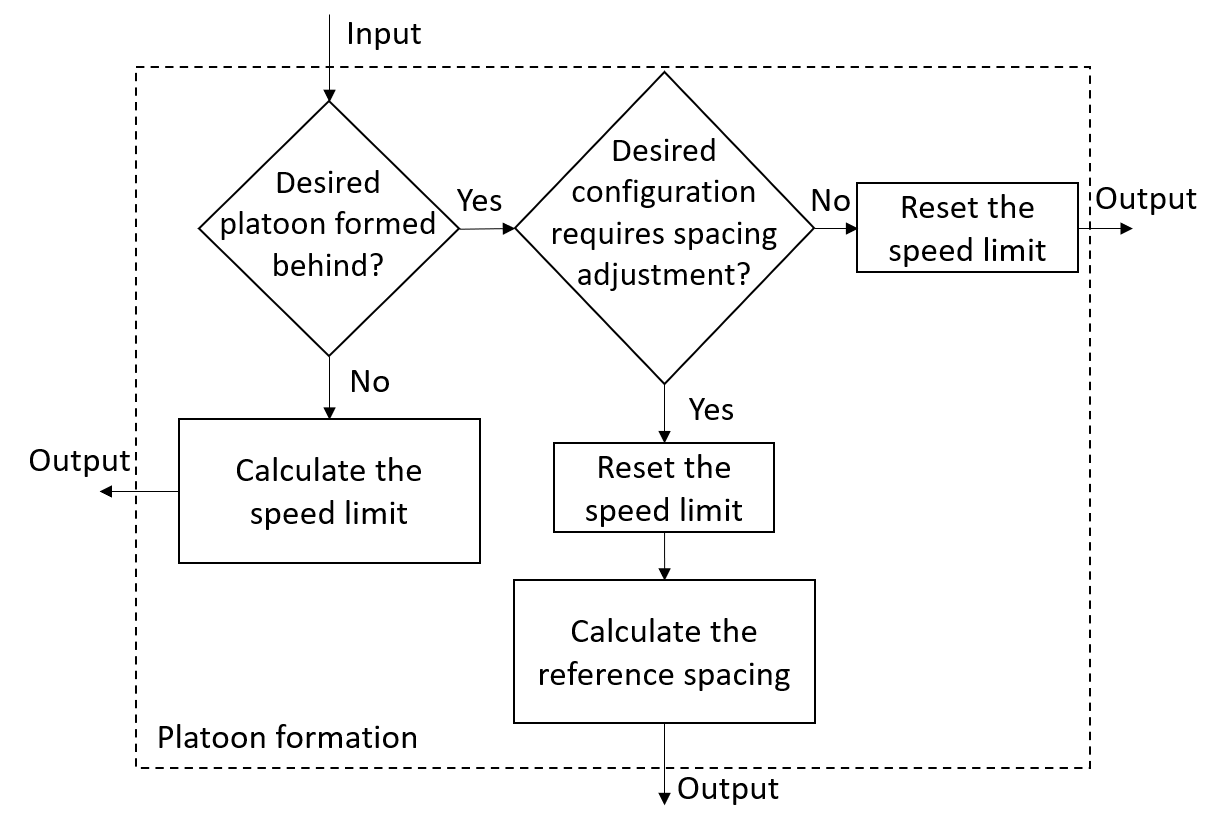}
    \vspace{0.2 cm}
    \caption{\sf  Logic diagram of a desired leader for creating the desired platoon}
    \label{Fig: Speed selection logic}
\end{figure}
Note that in Section \ref{Section: No Coordination}, vehicles used homogeneous time headway constants $h$, and constant desired speed $V_f$ in the cruise control mode. However, with the coordinator in the loop the desired time headways are (possibly) heterogeneous and time-varying, and the reference speed in the cruise control mode is, in general, piece-wise constant. Therefore, additional analysis is required in order to establish stability. 
\begin{theorem}\label{Theorem: (COOR) convergence to the desired configuration}
There exist design parameters such that the longitudinal controller \eqref{Eq: (Acceleration Control) Design of c_i}-\eqref{Eq: Switching distance Delta_d} with the reference speed/spacing specified in \eqref{Eq: Effect of logic B_V on ref speed}, \eqref{Eq: Effect of logic B_D on ref spacing}, guarantees that the system of vehicles converges to the configuration specified by the coordinator.
\end{theorem}
\begin{proof}
Refer to Appendix \ref{Section: Appendix B}.
\end{proof}
\begin{remark}
Consider the $m$-platoon symmetrical desired configuration, $1 < m \leq \frac{n}{2}$. Then, the equilibrium distance between adjacent platoons is $d = \frac{n}{m}(\frac{P}{n} - L - hV_f - S_0) + hV_f + S_0$. Since $P > n(L + hV_f + S_0)$, i.e., $n < n_c$, $d$ decreases when $m$ is increased from $2$ to $\frac{n}{2}$. In other words, as the size of the platoon increases in the $m$-platoon symmetrical configuration, the desired inter-platoon relative spacing increases. Since the sensing range of vehicles are limited, V2V communication allows accurate tracking of the position of the vehicle ahead even when it is outside of the sensing range. In other words, V2V communication expands the number of achievable desired configurations when there is central coordination.
\end{remark}
\begin{figure}[htb!]
\captionsetup[subfigure]{labelformat=empty}
    \begin{subfigure}{1\textwidth}
        \centering
    \includegraphics[width=0.50\textwidth]{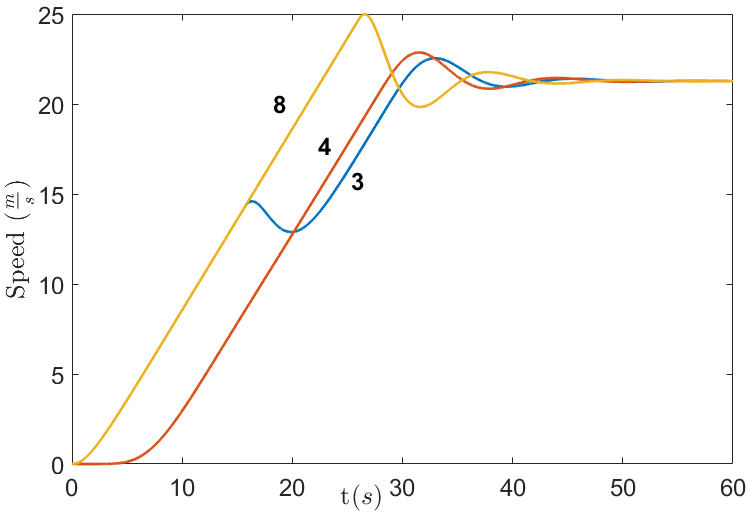}
  \caption{}
    \end{subfigure}
\vspace{-0.2cm}
    \begin{subfigure}{1\textwidth}
        \centering
    \includegraphics[width=0.50\textwidth]{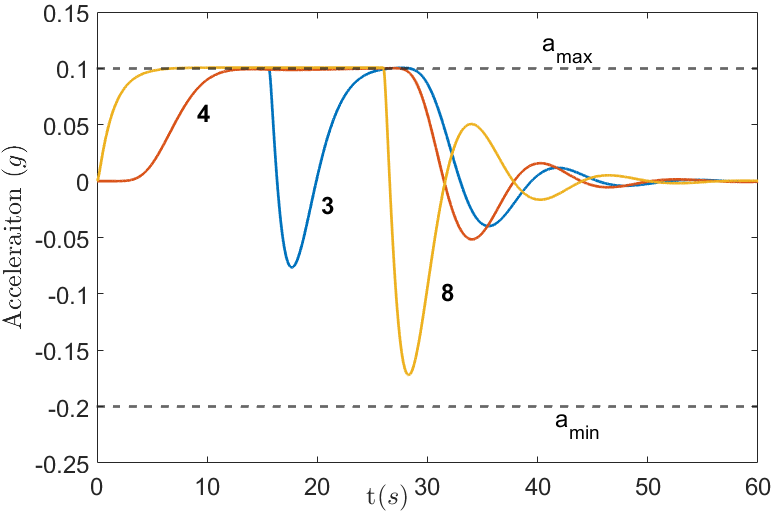}
  \caption{}
    \end{subfigure}
  \vspace{-0.2cm}
      \begin{subfigure}{1\textwidth}
        \centering
    \includegraphics[width=0.50\textwidth]{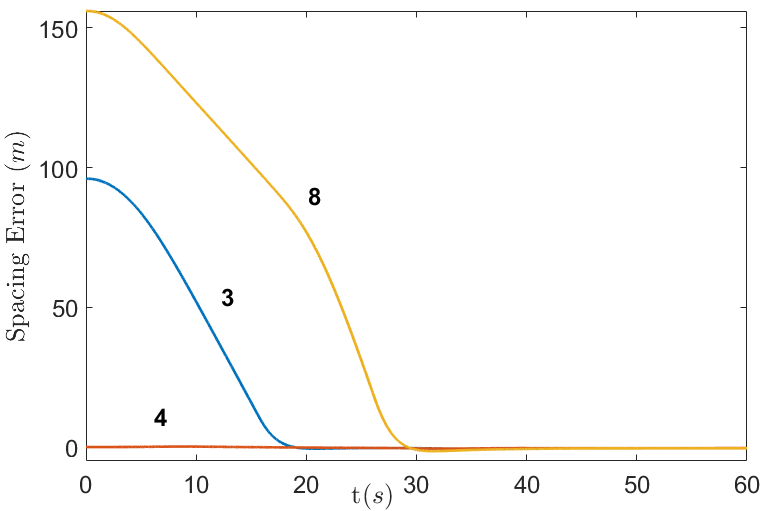}
  \caption{}
    \end{subfigure}
  \caption{\sf Simulation results for the high density traffic regime}\label{Fig: HighDensTraff}
\end{figure}
\begin{figure}[htb!]
\captionsetup[subfigure]{labelformat=empty}
    \begin{subfigure}{1\textwidth}
        \centering
    \includegraphics[width=0.50\textwidth]{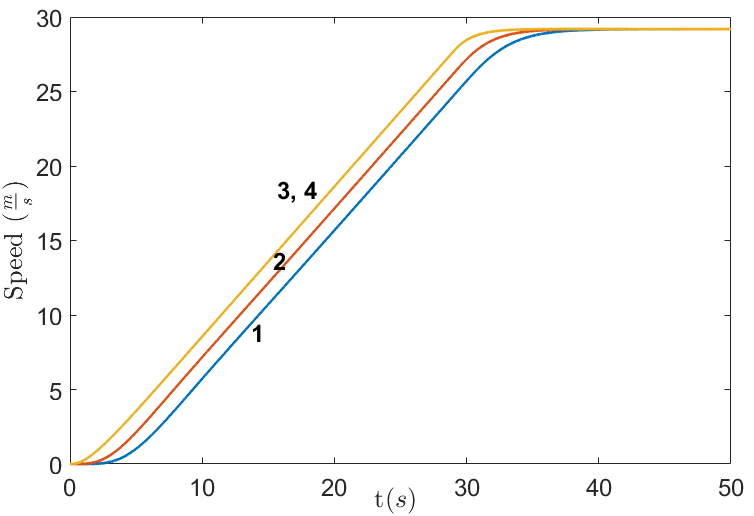}
  \caption{}
    \end{subfigure}
\vspace{-0.2cm}
    \begin{subfigure}{1\textwidth}
        \centering
    \includegraphics[width=0.50\textwidth]{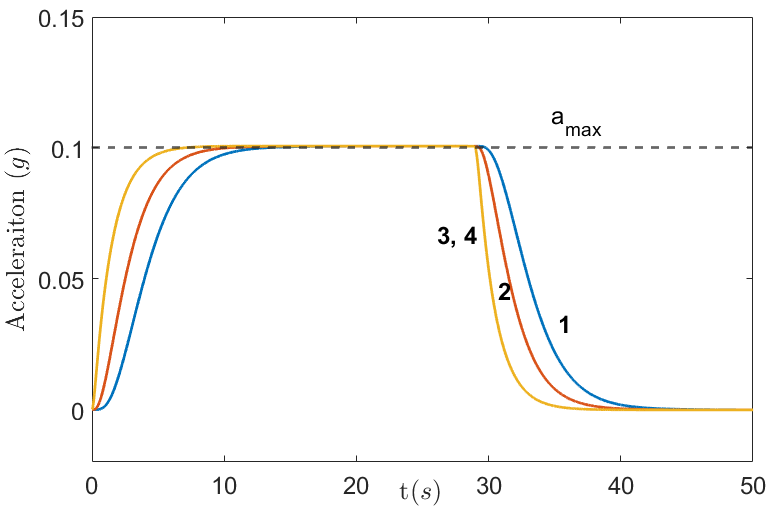}
  \caption{}
    \end{subfigure}
  \vspace{-0.2cm}
      \begin{subfigure}{1\textwidth}
        \centering
    \includegraphics[width=0.50\textwidth]{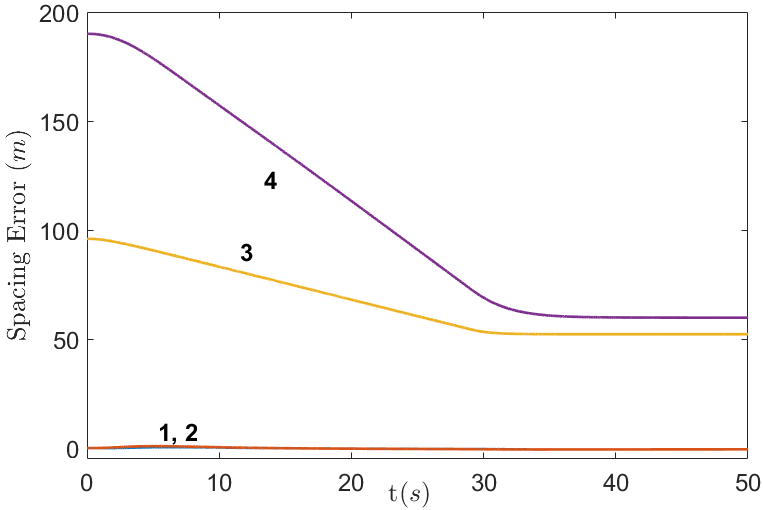}
  \caption{}
    \end{subfigure}
  \caption{\sf Simulation results for the low density traffic regime with no coordination}\label{Fig: LowDensTraff}
\end{figure}
\begin{figure}[htb!]
\captionsetup[subfigure]{labelformat=empty}
    \begin{subfigure}{1\textwidth}
        \centering
    \includegraphics[width=0.50\textwidth]{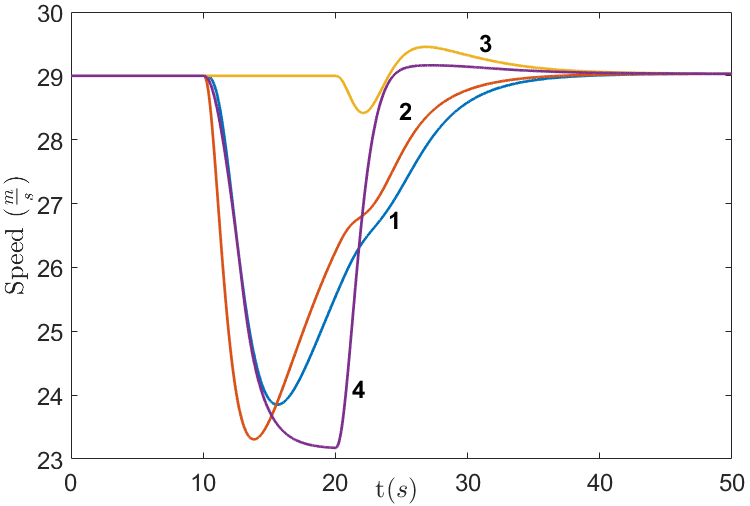}
  \caption{}
    \end{subfigure}
\vspace{-0.2cm}
    \begin{subfigure}{1\textwidth}
        \centering
    \includegraphics[width=0.50\textwidth]{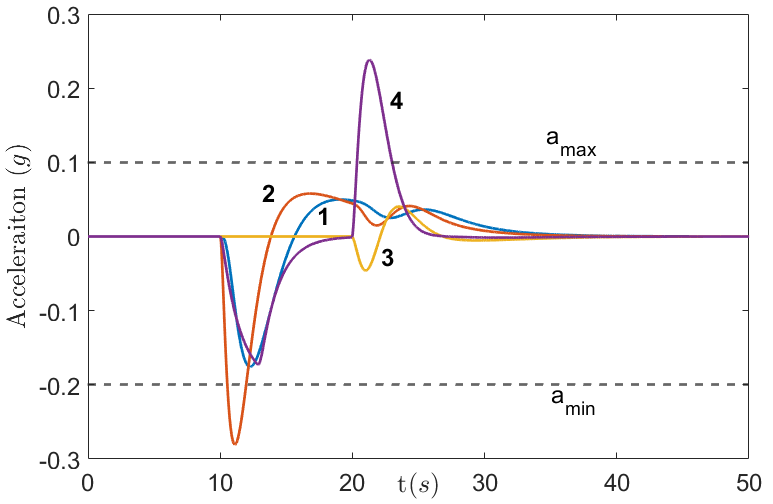}
  \caption{}
    \end{subfigure}
  \vspace{-0.2cm}
      \begin{subfigure}{1\textwidth}
        \centering
    \includegraphics[width=0.50\textwidth]{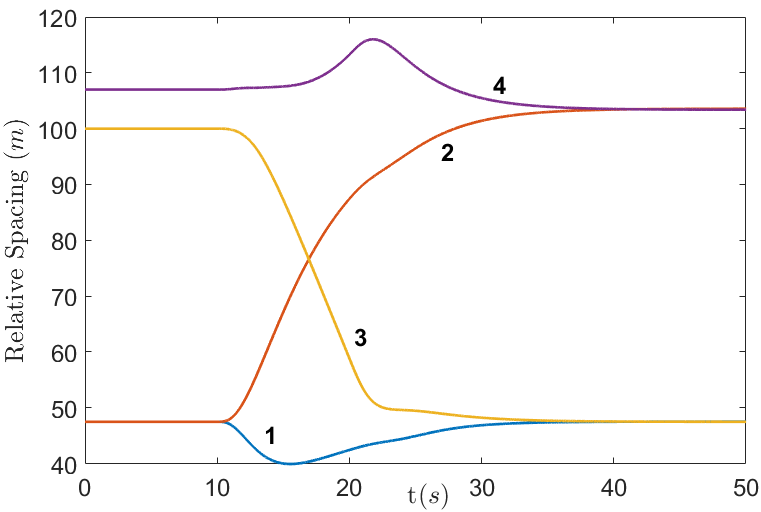}
  \caption{}
    \end{subfigure}
  \caption{\sf Simulation results with coordination and $2$-platoon symmetrical desired configuration}\label{Fig: LowDensTraff_COOR}
\end{figure}

\section{Simulation Results}\label{Section: Simulation}
In this section, we illustrate the performance of the designed control laws by simulating a few scenarios. In all scenarios, the control parameters are chosen as follows,
\begin{equation}\label{Eq: (Simulation)(Acceleration Control) Conditions to gurantee stability and string attenuation}
\begin{aligned}
    K&_a = -9, ~~ C_{p} = 2, ~~ C_{v} = 6, ~~ C_{q} = 0.01, ~~ C_{s} = 0.03 \\
    h &= 1.5 ~[s], ~~ S_0 = 4~ [m], ~~ p = 10, ~~ a_{min} = -0.2g \\
    a&_{max} = 0.1g, ~~ r = 1, ~~ \lambda = 0.5
\end{aligned}
\end{equation}
\par
For this choice of design constants, it can be checked that stability, string error attenuation, and comfort conditions are satisfied. Other parameters are chosen to be as follows,
\begin{equation*}
        P = 320 ~[m], ~~ L = 4.5 ~[m], ~~ V_f = 29 ~ [\frac{m}{s}]
\end{equation*}
\par
Therefore, the critical number of vehicles is $n_{c} = \frac{320}{1.5\times29+4+4.5} = 6.04$. 
\par
\begin{figure}[t]
    \begin{subfigure}{.5\textwidth}
        \centering
        \includegraphics[width=0.7\linewidth]{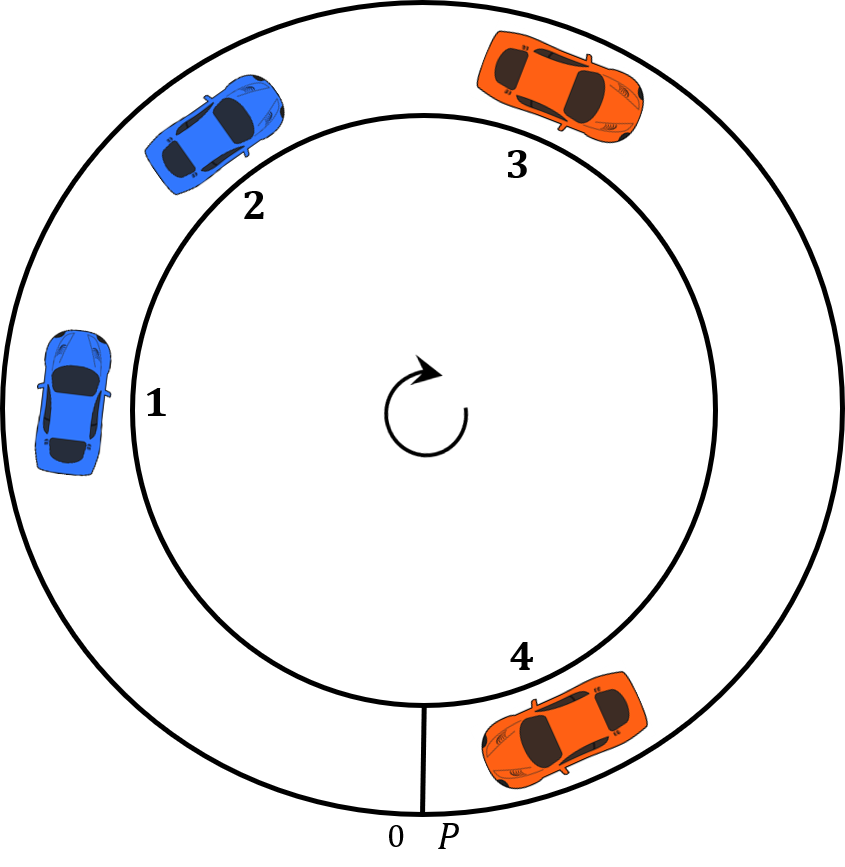}
        \caption{} \label{Fig: config init}
    \end{subfigure}
    \vspace{0.1cm}
    \begin{subfigure}{.5\textwidth}
        \centering
        \includegraphics[width=0.7\linewidth]{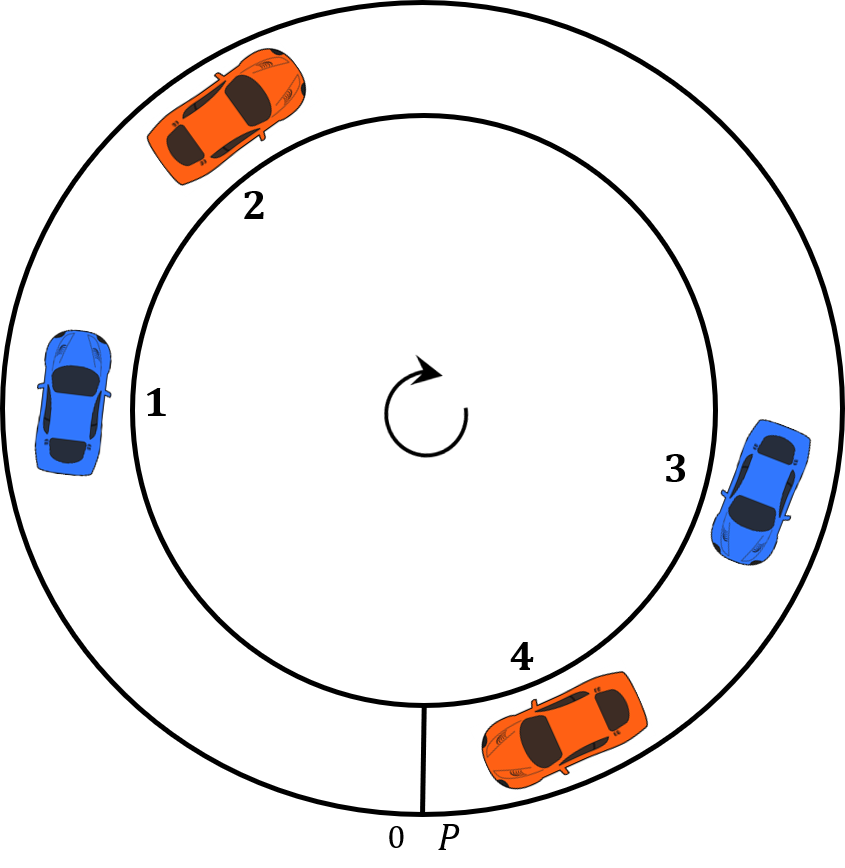}
        \caption{} \label{Fig: config cnst time hdwy}
    \end{subfigure}
        \vspace{0.1cm}
    \caption{\sf Steady state configuration of vehicles when there is (a) no coordination, (b) coordination with $2$-platoon symmetrical desired configuration, where vehicles in the cruise control mode are colored in orange and the ones in the vehicle following are colored blue}\label{Fig: Qualitative configurations}
\end{figure}
\par
\subsection{High Density Traffic Regime}
Let $n = 8 > n_ c$, i.e., a high-density traffic regime, with two platoons of sizes $3$ and $5$ initially at rest. The first platoon consists of vehicles $1 - 3$ with vehicle $3$ as the leader, and the second platoon consists of vehicles $4 - 8$ with vehicle $8$ as the leader. The distance between the first and second platoons is initially $100$ meters, i.e., $y_3(0) = 100~[m]$, and all the following vehicles are assumed to be at the desired spacing at $t = 0$. According to Theorem \ref{Theorem: (Acceleration Control) Robust veh following and speed tracking}, the system of vehicles converges to a unique configuration with the equilibrium relative spacing of $\frac{P}{n} - L = 35.5~[m]$, and speed of $\frac{1}{h}(\frac{P}{n}-S_0 - L) = 21~ [\frac{m}{s}]$. The speed, acceleration, and spacing error profiles for sample vehicles are shown in Figures \ref{Fig: HighDensTraff}. As can be seen from the acceleration and speed profiles, the leader of the first platoon, i.e., vehicle $3$, operates in the cruise control mode until $t \approx 15 \:\: [s]$. At this point it switches to the vehicle following mode and the two platoons become connected. Furthermore, the leader of the second platoon switches to the vehicle following mode at $t \approx 26 \:\: [s]$, and a platoon with no leader is formed. It is clear that the speed and acceleration profiles are smooth and within the comfort range in both modes of operation as well as during the transition between the two modes.
\subsection{Low Density Traffic Regime: No Coordination}
In this scenario, let $n = 4 < n_{c}$, i.e., a low-density traffic regime, and all vehicles are initially at rest. We assume that vehicles $1 - 3$ are initially in platoon formation with the following vehicles at desired spacing, and vehicle $4$ is $100$ meters ahead. The speed, acceleration, and spacing error profiles of the vehicles are shown in Figure \ref{Fig: LowDensTraff}. It can be seen from the spacing error profiles that vehicles $3$ and $4$ operate in the cruise control mode at all times, thus at steady state, we have one platoon of three vehicles with vehicle $3$ as the leader, and a single vehicle, i.e., vehicle $4$, operating in the cruise control mode, see Figure \ref{Fig: config init}. Moreover, it can be easily verified that all vehicles reach the free flow speed $V_f$ while satisfying the the acceleration bounds during the transient. Also, it can be seen from the acceleration profiles that the errors in acceleration are not magnified upstream a platoon, i.e., string error attenuation. 
\subsection{Low Density Traffic Regime: With Coordination}
We again consider $n = 4$. We assume that vehicles are travelling at steady state speed of $V_f$ and initial configuration of the previous scenario. Let the coordinator's desired configuration be $2$-platoon symmetrical with vehicles $2$ and $4$ as the desired leaders. The simulation results for this scenario are shown in Figure \ref{Fig: LowDensTraff_COOR}. The coordinator communicates the desired configuration to vehicles $2$ and $4$ at $t = 10\:\:[s]$. Since the platoon consisting of vehicles $3$ and $4$ has not formed at $t = 10\:\:[s]$, vehicle $4$ sets its speed limit to $V_s = 0.8V_f$, and starts to decelerate until vehicle $3$ catches up. At the same time vehicle $2$ smoothly increases its time headway constant to the desired value $h_d \approx 3.43\:\:[s]$, and starts decelerating in order to adjust its relative spacing. At $t \approx 20 \:\:[s]$, vehicle $3$ switches to the vehicle following mode, and vehicle $4$ resets its speed limit while smoothly increasing its desired relative spacing to $h_d v_1 + S_0$. Due to large initial positive relative spacing and speed error at $t \approx 20 \:\:[s]$, this introduces acceleration outside the comfortable bounds for vehicle $4$. It can be seen from the relative spacing profiles at $t \approx 35 \:\:[s]$, that the desired configuration is achieved asymptotically. The qualitative steady state configuration of this scenario is shown in Figure \ref{Fig: config cnst time hdwy}.
\section{Conclusion} \label{Section: Future Work}
In this paper we considered the design of vehicle longitudinal controllers for homogeneous vehicles following a single lane in a closed ring road under safety and comfort constraints in order to evaluate the impact of limited space on the speed of flow. We showed that if the number of vehicles is less than a certain critical number $n_c$, which depends on the perimeter of the ring road, free flow speed limit, and safety spacing, the vehicles can organize themselves around the ring road in an infinite number of different configurations. When the number of vehicles increases to be greater than or equal to $n_c$, all vehicles converge to a unique equilibrium configuration where the equilibrium speed decreases as the number of vehicles increases. When we add vehicle to vehicle communications, the controller is modified for faster action during vehicle following and safety can be guaranteed under lower intervehicle spacing. As a result, the critical number of vehicles $n_c$ that can operate at the maximum allowable speed increases. We also show that if a central coordinator dictates the configuration of the vehicles around the ring road, the proposed controllers can force the vehicles to converge to the desired coordination while respecting the safety and passenger comfort constraints. Computer simulations are used to demonstrate the performance of the controllers.
\appendices
\section{Proof of Theorem \ref{Theorem: (Acceleration Control) Robust veh following and speed tracking}}\label{Section: Appendix A}
\begin{enumerate}[label=(\roman*)]
    \item Let $w_e = \int_{0}^{t}[C_{q}(\tau)\delta_{e} + C_{s}(v_{r} - v_e)]d\tau$, where $C_q(t) = 0$ when the ego vehicle is in the cruise control mode. We consider the following three cases. First, let the ego vehicle be operating in the cruise control mode. The closed loop dynamics of the ego vehicle can be written as follows,
\begin{equation}\label{Eq: (Appx A) closed loop dynamics in cruise}
    \begin{aligned}
        \dot{v}_{e} &= a_e \\
        \dot{a}_e &= K_a a_e + C_{v}(v_{r} - v_e) + w_e \\
        \dot{w}_{e} &= C_{s}(v_{r} - v_{e}) \\
        \dot{v}_{r} &= \text{sat}[p(V_f - v_{r})]
    \end{aligned}
\end{equation}
    In an equilibrium point of \eqref{Eq: (Appx A) closed loop dynamics in cruise} we have, 
    \begin{equation}\label{Eq: Equilibrium analysis when I_i = 0}
    \begin{aligned}
          & a_{e} = 0 \\
          & C_{v}(v_{r} - v_{e}) + w_e = 0 \\
          & C_s(v_{r} - v_{e}) = 0 \\
          & v_{r} = V_f
    \end{aligned}
\end{equation}
and $(v_{e}, a_{e}, w_{e}, v_r) = (V_f, 0, 0, V_f)$ is the unique equilibrium of this mode. Without loss of generality, suppose that $p(V_f - v_{r}(0)) > a_{max}$, i.e., the saturation function is initially active. It follows that,
\begin{equation}\label{Eq: (Appx A) Time trajectory of V_r in cruise}
    v_{r}(t) =  \begin{cases*}
        a_{max}t & \mbox{if } $0 \leq t \leq T$ \\
        -\frac{a_{max}}{p}e^{-p(t - T)} + V_f  & \mbox{if } $T \leq t$
     \end{cases*}
\end{equation}
where $T = \frac{V_f}{a_{max}} - \frac{1}{p}$. Hence, $v_r \rightarrow V_f$ exponentially fast as $t \rightarrow \infty$. As a result, the state of the saturation function converges to its linear region and \eqref{Eq: (Appx A) closed loop dynamics in cruise} becomes a LTI system after a finite time. By shifting the equilibrium point of \eqref{Eq: (Appx A) closed loop dynamics in cruise} to zero and taking Laplace transform of the first three equations assuming zero initial condition, it follows that,
\begin{equation}\label{Eq: (Appx A) Cruise laplace relation between speeds}
    V_e(s) = K(s) V_r(s)
\end{equation}
where $V_e(s), V_r(s)$ are Laplace transforms of $v_e, v_r$, respectively, and,
\begin{equation}\label{Eq: Tf K(s)}
    K(s) = \frac{C_vs + C_s}{s^3 - K_a s^2 + C_v s + C_s}
\end{equation}
\par
For stability of the equilibrium, we require that poles of $K(s)$ lie in the open left half of the complex plane. For analyzing the performance in achieving comfort, note that from \eqref{Eq: (Appx A) Cruise laplace relation between speeds} we obtain,
\begin{equation}\label{Eq: (Appx A) Cruise laplace relation for acceleration}
    A_e(s) = sK(s)V_r(s)
\end{equation}
where $A_e(s)$ is the Laplace transform of $a_e$. Therefore, if we choose the design constants such that $|K(j\omega)| \leq 1$, $\forall \omega \geq 0$ and $k(t) \geq 0$, $\forall t \geq 0$, then $||k(t)||_{1} \leq 1$. Thus, assuming zero initial condition, it follows from \eqref{Eq: (Appx A) Time trajectory of V_r in cruise} and \eqref{Eq: (Appx A) Cruise laplace relation for acceleration} that for all $t \geq 0$,
\begin{equation}\label{Eq: (Appx A) Cruise relation for acceleration in time}
    |a_e(t)| \leq ||k(t)||_1 \sup_{t \geq 0}|\dot{v}_{r}(t)| \leq a_{max}
\end{equation}
\par
Similarly, if the saturation function is initially active in the other direction, i.e., $p(V_f - v_{r}(0)) < a_{min}$, then $a_e(t) \geq a_{min}$. 
\par
Next, let the ego vehicle switch to the vehicle following mode at time $t = 0$. The closed loop dynamics of the ego vehicle can be written as follows,
\begin{equation}\label{Eq: (Appx A) closed loop dynamics in follow}
\begin{aligned}
    \dot{y}_{e} &= v_{l} - v_{e} \\
    \dot{v}_{e} &= a_e \\
    \dot{a}_e &= 
    K_a a_e + C_{p}(t)\delta_e + C_{v}(v_{l} + (v_r(0) - v_l)e^{-\lambda t} - v_e) + w_e \\
    \dot{w}_{e} &= C_{q}(t)\delta_e + C_{s}(v_{l} + (v_r(0) - v_l)e^{-\lambda t} - v_{e})
\end{aligned}
\end{equation}
\par
Assuming constant $v_l$ and neglecting the exponentially vanishing terms, the equilibrium point of \eqref{Eq: (Appx A) closed loop dynamics in follow} is,
\begin{equation}\label{Eq: Equilibrium analysis when I_i = 1}
    \begin{aligned}
          & v_{e} = v_l\\
          & a_{e} = 0 \\
          & C_{p}(t)(y_{e} - hv_{e} - S_0) + C_{v}(v_{l} - v_{e}) + w_e = 0 \\
          & C_q(t)(y_{e} - hv_{e} - S_0) + C_s(v_{l} - v_e) = 0
    \end{aligned}
\end{equation}
\par
Therefore, $(y_e, v_e, a_e, w_e) = (hv_e +S_0, v_l, 0, 0)$ is the unique equilibrium of \eqref{Eq: (Appx A) closed loop dynamics in follow}. Let $z^T_e = \begin{pmatrix} y_{e} & v_{e} & a_{e} & w_{e}\end{pmatrix}$. Note that since $v_r(0)e^{-\lambda t} \rightarrow 0$ exponentially fast as $t \rightarrow \infty$, it has no effect on the stability and is ignored in the analysis that follows. By shifting the equilibrium to zero, \eqref{Eq: (Appx A) closed loop dynamics in follow} can be written in the following compact form,
\begin{equation}\label{Eq: Dynamics of vehicle m-1 in a platoon with leader}
    \dot{z}_{e} = (A_1 + D_{1}(t)) z_{e} + (B_1 + D_2(t)) v_{l}
\end{equation}
where,
\begin{equation}\label{Eq: Matrices A & B of dynamics of vehicle m-1 in a platoon with leader}
\begin{aligned}
    A_1 &=
  \begin{pmatrix}
    0 & -1 & 0 & 0\\
    0 & 0 & 1 & 0 \\
    C_{p} & -(hC_{p} + C_{v}) & K_a & 1 \\
    C_{q} & -(hC_{q} + C_{s}) & 0  & 0
  \end{pmatrix}, \:\:
  B_1 = \begin{pmatrix}
  1 \\
  0 \\
  C_v \\ 
  C_s
  \end{pmatrix}, \\
  D_{1}(t) &= \begin{pmatrix}
    0 & 0 & 0 & 0 \\
    0 & 0 & 0 & 0  \\
    -C_p e^{-\lambda t} & C_p h e^{-\lambda t} & 0 & 0 \\
    -C_q e^{-\lambda t} & C_q h e^{-\lambda t} & 0  & 0
  \end{pmatrix}, \:\:
    D_2(t) = \begin{pmatrix}
  0 \\
  0 \\
  -C_v e^{-\lambda t} \\ 
  - C_s e^{-\lambda t}
  \end{pmatrix}
\end{aligned}
\end{equation}
Since $D_{1}(t), D_{2}(t) \rightarrow 0$ as $t \rightarrow \infty$, the equilibrium of \eqref{Eq: Dynamics of vehicle m-1 in a platoon with leader} is exponentially stable if the equilibrium of the LTI system $\dot{z}_e = A_1z_e + B_1v_l$ is exponentially stable \cite{khalil2002nonlinear}. By taking Laplace transform of the corresponding LTI system, we arrive at the following relationship, 
\begin{equation}\label{Eq: (Appx A) Follow Laplace relationship between speeds}
    V_e(s) = G(s)V_l(s) + E_0(s)
\end{equation}
where $E_0(s)$ is due to non-zero initial condition of the ego vehicle and,
\begin{equation}\label{Eq: (Acceleration Control) G(s)}
\begin{aligned}
     G(s) &= \frac{C_{v} s^2 + (C_{p} + C_{s})s + C_{q}}{F(s)} \\
     F(s) &= s^4 - K_a s^3 + (hC_{p} + C_{v})s^2 + (C_{p} + hC_{q} + C_{s})s + C_{q}
\end{aligned}
    \end{equation}
\begin{figure}[t]
    \centering
    \includegraphics[width = 0.7\textwidth]{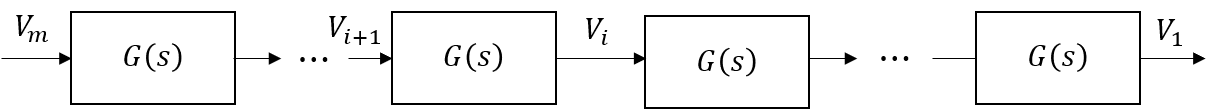}
    \vspace{0.2 cm}
    \caption{\sf Block diagram of a platoon of $m$ vehicles with vehicle $m$ as the leader}
    \label{Fig: line stability}
\end{figure}
\begin{figure}[t]
    \centering
    \includegraphics[width = 0.7\textwidth]{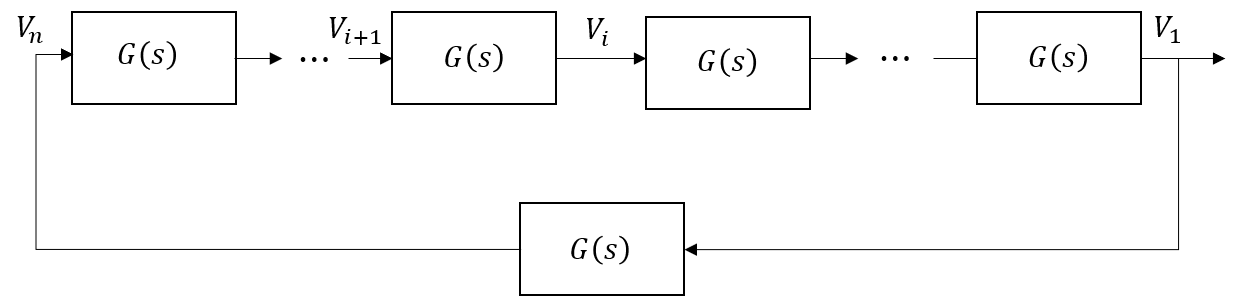}
    \vspace{0.2 cm}
    \caption{\sf Block diagram of a platoon with no leader}
    \label{Fig: ring stability}
\end{figure}
\par
We consider two possible cases. In the first case, the ego vehicle has joined a platoon of $m$ vehicles, $m \in \mathcal{N}$, with the leader in the cruise control mode, see Figure \ref{Fig: line stability}. In this case, the stability of the equilibrium is guaranteed if poles of $G(s)$ lie in the open left half of the $s$-plane. In the second case, the ego vehicle has joined a platoon with no leader, see Figure \ref{Fig: ring stability}. It is well-known \cite{khalil2002nonlinear} that a sufficient condition for exponential stability of the equilibrium of the system in Figure \ref{Fig: ring stability} is that poles of $G(s)$ lie in the open left half of the $s$-plane and $|G(j\omega)| \leq 1$, $\forall \omega \geq 0$. \par
For analyzing the performance of the controller in achieving passenger comfort, we assume that $C_p(t)$, $C_q(t)$ are constant and $v_r = v_l$ in the vehicle following mode in order to make use of the properties of LTI systems. However, we demonstrate by simulations that the proposed controller achieves smooth and comfortable acceleration during transient even if these assumptions do not hold. It follows from \eqref{Eq: (Appx A) Follow Laplace relationship between speeds} that,
\begin{equation}\label{Eq: Relatioship between A_i(s) and A_i+1(s)}
    A_e(s) = G(s)A_{l}(s) + \tilde{E}_0(s)
\end{equation}
where $G(s)$ is specified in \eqref{Eq: (Acceleration Control) G(s)} and $\tilde{E}_0(s) = sE_0(s) + G(s)v_l(0) - v_e(0)$ is due to non-zero initial condition. Since $|G(j\omega)| \leq 1$, $\forall \omega \geq 0$ in order to guarantee stability, if we choose the design constants such that $g(t) \geq 0$, $\forall t \geq 0$, we have,
\begin{equation}\label{Eq: (Appx A) VehFol relation for acceleration in time}
    |a_e(t)| \leq |a_{l}(t)| + \tilde{e}_0(t)
\end{equation}
where $\tilde{e}_0(t)$ is the inverse Laplace transform of $\tilde{E}_0(s)$ and is exponentially vanishing. Thus, the the following vehicles accelerate/decelerate at most as high as the vehicle ahead except, maybe, for an exponentially vanishing term. 
\par
We now analyze the attenuation of errors in, e.g., the position, of a vehicle upstream the platoon. Consider a platoon of $m$ vehicles, $m \in \mathcal{N}$, travelling with a constant speed such that for $i = 1, 2, \cdots, m-1$, $\delta_i(0) = 0$. Assume that there is a small perturbation in the position of the lead vehicle, i.e. vehicle $m$. Since we are assuming a small perturbation, we can neglect the dynamics of the acceleration limiter filter. By assuming $C_p(t) = C_p$, $C_q(t) = C_q$, $v_r = v_l$, for all of the following vehicles and taking Laplace transform of the equations in \eqref{Eq: (Appx A) closed loop dynamics in follow} we obtain,
\begin{equation}
    \frac{\Delta_{i}}{\Delta_{i+1}}(s) = G(s), \:\: i = 1, 2, \cdots, m-2
\end{equation}
where $\Delta_i(s), \Delta_{i+1}(s)$ are the Laplace transforms of $\delta_i, \delta_{i+1}$, respectively, and $G(s)$ is given in \eqref{Eq: (Acceleration Control) G(s)}. The necessary and sufficient condition for string error attenuation in the $\mathcal{L}_2$ sense is that $|G(j\omega)| \leq 1$, $\forall \omega \geq 0$ \cite{Ioannou.Chien.1993}. Note that this condition is already satisfied in order to ensure stability of the equilibrium. Moreover, a sufficient condition for string error attenuation in the $\mathcal{L}_{\infty}$ sense is that $|G(j\omega)| \leq 1$, $\forall \omega \geq 0$, and $g(t) \geq 0$, $\forall t \geq 0$ \cite{Ioannou.Chien.1993}. This condition is also satisfied in order to provide comfort. Finally, due to homogeneity of vehicles, string error attenuation extends to the speed and acceleration errors (see also \eqref{Eq: (Appx A) Follow Laplace relationship between speeds} and \eqref{Eq: Relatioship between A_i(s) and A_i+1(s)}). \par
\item 
Let $n < n_c$, or equivalently $P > n(hV_f + S_0 + L)$. We claim that at least one vehicle must be operating in the cruise control mode at steady state. Suppose not; then from the equilibrium analysis in \eqref{Eq: Equilibrium analysis when I_i = 1} it follows that at steady state we have for every $i \in \mathcal{N}$ that $y_{i} = hv_i + S_0$, $v_{i} = v_{i+1}$. Therefore, $y_{i} = hv_i + S_0 = hv_{i+1} + S_0 = y_{i+1}$. Since $\sum_{i = 1}^{n}y_i = P - nL$, we obtain for every $i \in \mathcal{N}$ that $y_{i} = \frac{P}{n} - L$ and,
\begin{equation*}
    v_{i} = \frac{1}{h}(\frac{P}{n} - S_0 - L) > V_f
\end{equation*}
which cannot occur because, according to the designed logic, this violates the speed limit $V_f$. It follows from the equilibrium analysis in \eqref{Eq: Equilibrium analysis when I_i = 0}, \eqref{Eq: Equilibrium analysis when I_i = 1} and the stability of the equilibrium from the previous part that for every $i \in \mathcal{N}$, if vehicle $i$ operates in the cruise control mode at steady state, its speed converges to $V_f$. Moreover, it is required from the switching logic that $\delta_i \geq 0$. Hence, the relative spacing of vehicle $i$ converges to $hV_f + S_i$, where $S_i \geq S_0$ depends on the initial condition. On the other hand, if vehicle $i$ operates in the vehicle following mode, its speed converges to $V_f$, and its relative spacing converges to $hV_f + S_0$.
\item 
Let $n \geq n_c$, or equivalently $P \leq n(hV_f + S_0 + L)$, then all vehicles must be operating in the vehicle following mode at steady state. Otherwise, using a similar argument as before we arrive at the contradiction $P > n(hV_f + S_0 + L)$. Hence, for every $i \in \mathcal{N}$, the relative spacing of vehicle $i$ converges to $\frac{P}{n} - L$, and its speed converges to $\frac{1}{h}(\frac{P}{n} - S_0 - L)$. 
\end{enumerate}
\section{Proof of Theorem \ref{Theorem: (COOR) convergence to the desired configuration}}\label{Section: Appendix B}
Since the constant term $S_0$ has no effect on the stability, we neglect it in the analysis whenever needed. We consider the following three cases: first, consider the $1$-platoon asymmetrical desired configuration and let the ego vehicle be the desired leader. At $t = 0$, the ego vehicle sets its speed limit to $\alpha V_f$, $\alpha \in (0,1)$, and starts tracking $\alpha V_f$. Using the equilibrium equations in \eqref{Eq: Equilibrium analysis when I_i = 0}, it follows that the equilibrium of the closed loop dynamics of the ego vehicle is $v_e = v_{r} = \alpha V_f$, $a_e = 0$, and $w_e = 0$. Since at the equilibrium we must have $v_{i} = v_e = \alpha V_f$, $i \in \mathcal{N}$, and $\alpha < 1$, it follows that no other vehicle can be operating in the cruise control mode at steady state. Therefore, all of the vehicles switch to the vehicle following mode after a finite time and form a platoon of $n$ vehicles with the ego vehicle as its leader. Using $\alpha V_f$ instead of $V_f$ in the analysis from \eqref{Eq: (Appx A) Time trajectory of V_r in cruise} - \eqref{Eq: (Appx A) Follow Laplace relationship between speeds}, exponential stability of the equilibrium of the platoon immediately follows. Similarly, the equilibrium state of the platoon is exponentially stable when the ego vehicle resets its speed limit to $V_f$. Note that the total number of switching in the reference speed and spacing is finite, thus the switching does not affect stability.  \par
Consider the symmetrical desired configuration. Note that by construction, all vehicles eventually switch to the vehicle following mode in order to adjust their relative spacing by using the desired time headway constant and do not switch again at future times (see Figure \ref{Fig: Speed selection logic}). Thus, the switching does not affect the steady state behavior of vehicles. Without loss of generality, let the ego vehicle be in the vehicle following mode when it smoothly increases its time headway constant from $h$ to $h_1$ at time $t = 0$, where $h_1$ is such that $h_1 V_f + S_0 = \frac{P}{n} - L$. The closed loop dynamics of the ego vehicle can be written as follows,
\begin{equation}\label{Eq: (COOR) linear set of equations in z}
\begin{aligned}
    \dot{y}_{e} &= v_{l} - v_{e} \\
    \dot{v}_{e} &= a_e \\
    \dot{a}_e &=  K_a a_e + C_{p}(t)(y_e - h(t)v_e)  + C_{v}(v_{l} - v_{e}) + w_{e} \\
    \dot{w}_{e} &= C_{q}(t)(y_e - h(t)v_e)  + C_{s}(v_{l} - v_{e})
\end{aligned}
\end{equation}
where $h(t) = h_1 + (h - h_1)e^{-\lambda t}$. The design parameters $C_p(t)$, $C_q(t)$ are also smoothly changed to desired values, e.g., $C_p(t) = \tilde{C}_p + (C_p - \tilde{C}_p)e^{-\lambda t}$, where $\tilde{C}_p > 0$ is a design constant to be chosen. 
Note that if the ego vehicle was operating in the cruise control mode at $t = 0$, its closed loop dynamics would have been the same as \eqref{Eq: (COOR) linear set of equations in z} except that $C_p(0) = C_q(0) = 0$, and $v_r = v_l + (v_r(0) - v_l)e^{-\lambda t}$. Let $z^T_e = \begin{pmatrix} y_{e} & v_{e} & a_{e} & w_{e}\end{pmatrix}$. 
We can write the closed loop dynamics \eqref{Eq: (COOR) linear set of equations in z} in the following compact form,
\begin{equation}\label{Eq: (COOR) Dynamics of vehicle 1 with symm config begining with veh fol}
    \dot{z}_e = (A_2 + D_{3}(t)) z_{e} + B_2 v_{l}
\end{equation}
where, 
\begin{equation*}\label{Eq: Matrix D_2(t) when vehicle 2 changes its time headway}
A_2 =
  \begin{pmatrix}
    0 & -1 & 0 & 0\\
    0 & 0 & 1 & 0 \\
    \tilde{C}_{p} & -(h_1\tilde{C}_{p} + C_{v}) & K_a & 1 \\
    \tilde{C}_{q} & -(h_1\tilde{C}_{q} + C_{s}) & 0  & 0
  \end{pmatrix}, \:\:
    B_2 = \begin{pmatrix}
  1 \\
  0 \\
  C_v \\ 
  C_s
  \end{pmatrix}
\end{equation*}
in which $D_3(t) \rightarrow 0$ exponentially fast as $t \rightarrow \infty$ and is not brought here for the sake of brevity. The equilibrium of \eqref{Eq: (COOR) Dynamics of vehicle 1 with symm config begining with veh fol} is exponentially stable if the equilibrium of the LTI system,
\begin{equation}\label{Eq: (COOR) LTI Dynamics of vehicle 1 with symm config}
    \dot{z}_e = A_2 z_e + B_2v_l
\end{equation}
is exponentially stable \cite{khalil2002nonlinear}. In order to find the stability condition, we use another representation of \eqref{Eq: (COOR) LTI Dynamics of vehicle 1 with symm config} by taking Laplace transform of both sides and deriving the following relationship,  
\begin{equation}\label{Eq: (COOR) Relationship between V_1 and V_2 in the symm config}
    V_e(s) = H_1(s)V_l(s) + E_0(s)
\end{equation}
where $E_0(s)$ is due to non-zero initial condition of the ego vehicle and,
\begin{equation*}
    H_1(s) = \frac{\tilde{C}_p s + \tilde{C}_q}{s^4 - K_a s^3 + (h_1 \tilde{C}_p + C_v)s^2 + (\tilde{C}_p + h_1 \tilde{C}_q + C_s)s + \tilde{C}_q}
\end{equation*}
is a transfer function similar to $G(s)$ in \eqref{Eq: (Acceleration Control) G(s)} only with different parameters. Since the desired configuration is symmetrical, i.e., $h_d = h_1$ for all vehicles, the closed loop dynamics of each vehicle on the ring road becomes the same as \eqref{Eq: (COOR) linear set of equations in z} after the final switching time only with (possibly) different initial values of the design parameters and reference speed. Hence, \eqref{Eq: (COOR) Relationship between V_1 and V_2 in the symm config} holds for all vehicles in the corresponding LTI system and we have a similar block diagram as in Figure \ref{Fig: ring stability}. Using a similar argument to the proof of Theorem \ref{Theorem: (Acceleration Control) Robust veh following and speed tracking}, a sufficient condition for stability is that poles of $H_1(s)$ lie in the left half of the $s$-plane and $|H_1(j\omega)| \leq 1$, $\forall \omega \geq 0$. It can be verified that if the design constants $K_a, C_p, C_q$ satisfy $K_a C_p + C_q < 0$, and $\tilde{C}_p$, $\tilde{C}_q$ are chosen such that $h_1 \tilde{C}_p = hC_p$, $h_1 \tilde{C}_q = hC_q$, the stability conditions are guaranteed.
\begin{figure}[t]
    \centering
    \includegraphics[width = 0.7\textwidth]{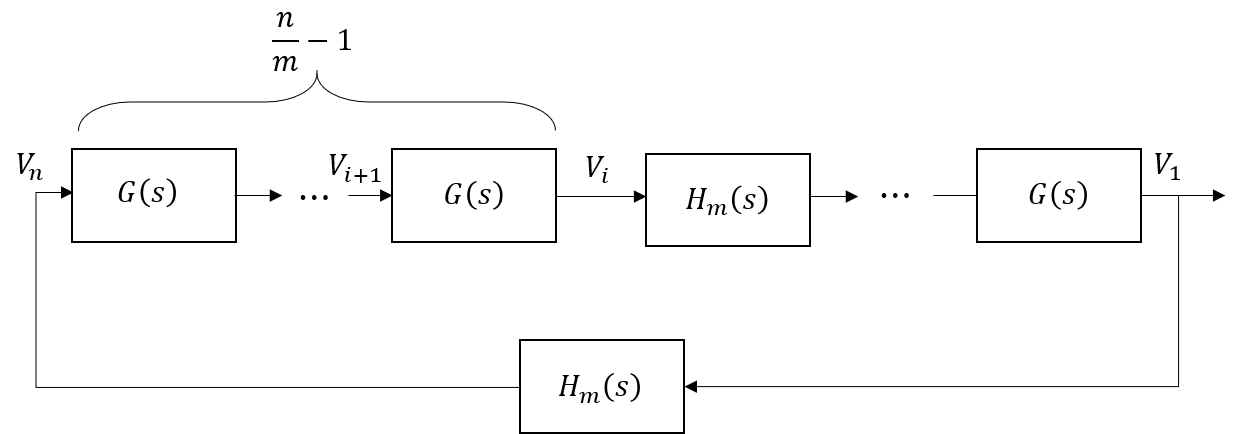}
    \vspace{0.2 cm}
    \caption{\sf Block diagram of the $m$-platoon symmetrical configuration with vehicle $n$ as the desired leader of a platoon}
    \label{Fig: ring stability in the m platoon configuration}
\end{figure}
\par
Finally, consider the $m$-platoon symmetrical desired configuration,  $1 < m \leq \frac{n}{2}$, and let the ego vehicle be a desired leader. From the result for the $1$-platoon asymmetrical configuration, it follows that after a finite time, all of the desired followers of the ego vehicle switch to the vehicle following mode. Also, similar to the symmetrical configuration, all vehicles eventually operate in the vehicle following mode and do not switch at future times. Assume that the ego vehicle sets its reference spacing to $h_m v_e + S_0$ at $t = 0$, where $h_m$ is the desired time headway constant at the free flow speed calculated by the ego vehicle. The closed loop dynamics of the ego vehicle can be written as \eqref{Eq: (COOR) linear set of equations in z} with a different $h_m$ and (possibly) different design constants and reference speed trajectory. Following a similar argument as in the previous scenario, the following relationship can be found for the corresponding LTI system,   
\begin{equation}\label{Eq: (COOR) Relationship between V_1 and V_2 in the m-platoon config}
        V_e(s) = H_m(s)V_l(s) + E_0(s)
\end{equation}
where $E_0(s)$ is due to non-zero initial condition and,
\begin{equation*}
    H_m(s) = \frac{\tilde{C}_p s + \tilde{C}_q}{s^4 - K_a s^3 + (h_m \tilde{C}_p + C_v)s^2 + (\tilde{C}_p + h_m \tilde{C}_q + C_s)s + \tilde{C}_q}
\end{equation*}
where the constants $\tilde{C}_p$, $\tilde{C}_q$ are to be chosen. Note that the Laplace domain relationship between the speeds of the following vehicles of a platoon is different than \eqref{Eq: (COOR) Relationship between V_1 and V_2 in the m-platoon config} and was derived in \eqref{Eq: (Appx A) Follow Laplace relationship between speeds}. Figure \ref{Fig: ring stability in the m platoon configuration} shows the block diagram of this case after the final switching time.
By using a similar argument as in the proof of Theorem \ref{Theorem: (Acceleration Control) Robust veh following and speed tracking}, a sufficient condition for stability is that poles of $H_m(s)G(s)$ lie in the open left half of the $s$-plane, and $|H_m(j\omega)G(j\omega)| \leq 1$, $\forall \omega \geq 0$. Since the transfer function $G(s)$ is stable and satisfies $|G(j\omega)| \leq 1$, $\forall \omega \geq 0$, a sufficient condition for stability is that $H_m(s)$ is stable and $|H_m(j\omega)| \leq 1$, $\forall \omega \geq 0$. These conditions are automatically guaranteed if $K_a C_p + C_q < 0$, and $h_m \tilde{C}_p = hC_p$, $h_m \tilde{C}_q = hC_q$. 


\bibliographystyle{ieeetr}
\bibliography{arXiv_Main_IV20}
\end{document}